\documentclass[letterpaper, 10 pt, conference]{ieeeconf}
\IEEEoverridecommandlockouts
\usepackage[english]{babel}
\usepackage{theorem}
\theoremheaderfont{\itshape\bfseries}
{\theorembodyfont{\itshape}
\newtheorem{assumption}{\textbf{Assumption}}

\newtheorem{definition}{\textbf{Definition}}
\newtheorem{theorem}{\textbf{Theorem}}

\newtheorem{remark}{\textbf{Remark}}

\newtheorem{problem}{\textbf{Problem}}
}
\usepackage{graphicx}
\usepackage{newtxtext,newtxmath}
\usepackage{amsmath}
\usepackage{amsfonts}
\usepackage{amssymb}
\usepackage{mathrsfs}
\usepackage{xcolor}
\usepackage{graphicx}

\newcommand{\T}{^{\mbox{\tiny T}}}
\newcommand{\R}{\mathbb{R}}

\newcommand{\eps}{\varepsilon}
\let\leq\leqslant
\let\geq\geqslant

\newenvironment{proof}[1][Proof]%
{\par\noindent\textit{#1:\ }}%
{\hspace*{\fill} \rule{6pt}{6pt}}
\newenvironment{proof*}[1][Proof]%
{\par\noindent\textit{#1:\ }}{}

\DeclareMathOperator{\diag}{diag}
\DeclareMathOperator{\trace}{tr}

\DeclareMathOperator{\sat}{sat}
\DeclareMathOperator{\sgn}{sgn}
\usepackage{tikz}
\usetikzlibrary{shapes,calc,arrows,patterns,decorations.pathmorphing
	,decorations.markings}
\usetikzlibrary{arrows.meta}

\newenvironment{system}[1]%
{\setlength{\arraycolsep}{0.5mm}\left\{ \; \begin{array}{#1}}%
    {\end{array} \right.}
\newenvironment{system*}[1]%
{\setlength{\arraycolsep}{0.5mm} \begin{array}{#1}}%
  {\end{array}}

\title{\LARGE \textbf{State Synchronization for Homogeneous Networks of Non-introspective Agents in Presence of Input Saturation-- A Scale-free Protocol Design}}

\author{Zhenwei Liu, Ali Saberi, Anton A. Stoorvogel, and Donya Nojavanzadeh%
  \thanks{Zhenwei Liu is with College of Information Science and
    Engineering, Northeastern University, Shenyang 110819,
    P. R. China {\tt\small jzlzwsy@gmail.com}} \thanks{Ali Saberi is with
    School of Electrical Engineering and Computer Science, Washington
    State University, Pullman, WA 99164, USA {\tt\small
      saberi@eecs.wsu.edu}} \thanks{Anton A. Stoorvogel is with
    Department of Electrical Engineering, Mathematics and Computer
    Science, University of Twente, P.O. Box 217, Enschede, The
    Netherlands {\tt\small A.A.Stoorvogel@utwente.nl}} \thanks{Donya Nojavanzadeh is with
    School of Electrical Engineering and Computer Science, Washington
    State University, Pullman, WA 99164, USA {\tt\small
    	donya.nojavanzadeh@wsu.edu}} }

\begin{document}

\maketitle
\thispagestyle{empty}
\pagestyle{empty}

\begin{abstract}
   This paper studies global and semi-global regulated state synchronization of homogeneous networks of non-introspective agents in presence of input saturation based on additional information exchange where the reference trajectory is given by a so-called exosystem which is assumed to be globally reachable. Our protocol design methodology does not need any knowledge of the directed network topology and the spectrum of the associated Laplacian matrix. Moreover, the proposed protocol is scalable and achieves synchronization for any arbitrary number of agents.
  \end{abstract}

\section{Introduction}

The synchronization problem of multi-agent systems (MAS) has attracted
substantial attention during the past decade, due to the wide potential for
applications in several areas such as automotive vehicle control,
satellites/robots formation, sensor networks, and so on. See for
instance the books \cite{ren-book} and \cite{wu-book} or the
survey paper \cite{saber-murray3}.

State synchronization inherently requires homogeneous networks
(i.e. agents which have identical dynamics). Therefore, in this paper we
focus on homogeneous networks. State synchronization based on diffusive \emph{full-state coupling} has been studied where the agent dynamics progress from single- and double-integrator dynamics (e.g.  \cite{saber-murray2}, \cite{ren}, \cite{ren-beard}) to more general dynamics (e.g. \cite{scardovi-sepulchre}, \cite{tuna1},
\cite{wieland-kim-allgower}). State synchronization based on
diffusive \emph{partial-state coupling} has also been considered, including static design (\cite{liu-zhang-saberi-stoorvogel-auto} and \cite{liu-zhang-saberi-stoorvogel-ejc}), dynamic design (\cite{kim-shim-back-seo}, \cite{seo-back-kim-shim-iet}, \cite{seo-shim-back}, \cite{su-huang-tac},
\cite{tuna3}), and the design with additional communication (\cite{chowdhury-khalil} and \cite{scardovi-sepulchre}). 

Meanwhile, it is worth to note that actuator saturation is pretty common and indeed is ubiquitous in engineering
applications. Some researchers have tried to establish the (semi) global state and output synchronization results for MAS in the presence of input saturation.
Global synchronization for neutrally stable
agents has
been studied by \cite{meng-zhao-lin-2013} (continuous-time) and
\cite{yang-meng-dimarogonas-johansson} (discrete-time) for either undirected or detailed balanced graph. Then, global synchronization via static protocols for MAS with partial state coupling and linear general dynamics is developed in \cite{liu-saberi-stoorvogel-zhang-ijrnc}. Reference \cite{li-xiang-wei} provides the design which can deal with networks that are not
detailed balanced but intrinsically requires the agents to be single integrators. Similar scenarios also appear in \cite{fu-wen-yu-ding} (finite-time consensus), and \cite{yi-yang-wu-johansson-auto} (event-triggered control). 

Semi-global leader-follower state synchronization has been studied in
\cite{su-chen-chen-2015} and \cite{su-chen-lam-lin} in the case of
full-state coupling. References \cite{su-chen},
\cite{su-jia-chen} and \cite{wang-su-wang-chen-2016} provide the semi global result via partial state coupling but they all
require extra communication and are introspective. Adaptive approach also is studied in \cite{chu-yuan-zhang} but the observer requires extra communication and is
introspective. A low gain design is introduced in \cite{shi-li-lin-ijrnc2018} for heterogeneous MAS with introspective agents and requires extra communication to track any trajectory from exosystem. The paper \cite{yang-stoorvogel-grip-saberi} considers non-introspective agents and requires extra communication for heterogeneous MAS, and \cite{zhang-chen-su} has similar design for discrete-time MAS. Then, \cite{takaba6} studied MAS with 
non-introspective agents and does not require extra communication, however it requires solution of a nonconvex
optimization problem to find a dynamic protocol. Recently, \cite{zhang2018semiglobal} studid the semi-global state synchronization of homogeneous networks for both continuous/discrete-time MASs with non-introspective agents with both full-state and partial-state coupling in the presence of input saturation.

In this paper, we deal with global and semi-global regulated state synchronization problems for MAS in presence of input saturation by tracking the trajectory of an exosystem. 
We design dynamic protocols by using additional information exchange for MAS with non-introspective agents and for both networks with full- and partial-state coupling.
The protocol design is scalable and does not need any information of communication network except connectivity. In other words, the proposed protocols work for any MAS with an arbitrary number of agents.

\subsection*{Notations and definitions}

Given a matrix $A\in \mathbb{R}^{m\times n}$, $A\T$ denotes the transpose of $A$ and $\|A\|$ denotes the induced 2-norm of $A$. For a vector $x\in \mathbb{R}^q$, $\|x\|$ denotes the 2-norm of $x$ and for a vector signal $v$, we denote the $\mathscr{L}_1$, $\mathscr{L}_2$, and $\mathscr{L}_\infty$ norm by $\| v \|_1$, $\| v \|_2$ and $\| v \|_\infty$ respectively. A square matrix $A$ is said to be Hurwitz stable if all its eigenvalues are in the open left half complex plane. We denote by
$\diag\{A_1,\ldots, A_N \}$, a block-diagonal matrix with
$A_1,\ldots,A_N$ as its diagonal elements. $A\otimes B$ depicts the
Kronecker product between $A$ and $B$. $I_n$ denotes the
$n$-dimensional identity matrix and $0_n$ denotes $n\times n$ zero
matrix; sometimes we drop the subscript if the dimension is clear from
the context.

To describe the information flow among the agents we associate a \emph{weighted graph} $\mathcal{G}$ to the communication network. The weighted graph $\mathcal{G}$ is defined by a triple
$(\mathcal{V}, \mathcal{E}, \mathcal{A})$ where
$\mathcal{V}=\{1,\ldots, N\}$ is a node set, $\mathcal{E}$ is a set of
pairs of nodes indicating connections among nodes, and
$\mathcal{A}=[a_{ij}]\in \mathbb{R}^{N\times N}$ is the weighted adjacency matrix with non negative elements $a_{ij}$. Each pair in $\mathcal{E}$ is called an \emph{edge}, where
$a_{ij}>0$ denotes an edge $(j,i)\in \mathcal{E}$ from node $j$ to
node $i$ with weight $a_{ij}$. Moreover, $a_{ij}=0$ if there is no
edge from node $j$ to node $i$. We assume there are no self-loops,
i.e.\ we have $a_{ii}=0$. A \emph{path} from node $i_1$ to $i_k$ is a
sequence of nodes $\{i_1,\ldots, i_k\}$ such that
$(i_j, i_{j+1})\in \mathcal{E}$ for $j=1,\ldots, k-1$. A \emph{directed tree} is a subgraph (subset
of nodes and edges) in which every node has exactly one parent node except for one node, called the \emph{root}, which has no parent node. The \emph{root set} is the set of root nodes. A \emph{directed spanning tree} is a subgraph which is
a directed tree containing all the nodes of the original graph. If a directed spanning tree exists, the root has a directed path to every other node in the tree.  

For a weighted graph $\mathcal{G}$, the matrix
$L=[\ell_{ij}]$ with
\[
\ell_{ij}=
\begin{system}{cl}
\sum_{k=1}^{N} a_{ik}, & i=j,\\
-a_{ij}, & i\neq j,
\end{system}
\]
is called the \emph{Laplacian matrix} associated with the graph
$\mathcal{G}$. The Laplacian matrix $L$ has all its eigenvalues in the
closed right half plane and at least one eigenvalue at zero associated
with right eigenvector $\textbf{1}$ \cite{royle-godsil}. Moreover, if the graph contains a directed spanning tree, the Laplacian matrix $L$ has a single eigenvalue at the origin and all other eigenvalues are located in the open right-half complex plane \cite{ren-book}.

\section{Problem formulation}

Consider a MAS consisting of $N$ identical dynamic agents with input saturation:
\begin{equation}\label{eq1}
\begin{cases}
\dot{x}_i=Ax_i+B\sigma(u_i),\\
y_i=Cx_i,
\end{cases}
\end{equation}
where $x_i\in\mathbb{R}^n$, $y_i\in\mathbb{R}^q$ and
$u_i\in\mathbb{R}^m$ are the state, output, and the input of agent 
$i=1,\ldots, N$, respectively. Meanwhile,
\[
\sigma(v)=\begin{pmatrix} 
\sat(v_1) \\ \sat(v_2) \\ \vdots \\ \sat(v_m)
\end{pmatrix}\quad\text{ where }\quad
v=\begin{pmatrix} 
v_1 \\ v_2 \\ \vdots \\ v_m
\end{pmatrix} \in \R^m
\]
with $\sat(w)$ is the standard saturation function:
\[
\sat(w)=\sgn(w)\min(1,|w|).
\]

\begin{assumption}\label{Aass}
	Assume agents are at most weakly unstable, namely, all eigenvalues of $A$ are in the closed left half plane. Moreover, let $(A, B, C)$ be stabilizable and detectable.
\end{assumption}

The network provides agent $i$ with the following information,
\begin{equation}\label{eq2}
\zeta_i=\sum_{j=1}^{N}a_{ij}(y_i-y_j),
\end{equation}
where $a_{ij}\geq 0$ and $a_{ii}=0$. This communication topology of
the network can be described by a weighted graph $\mathcal{G}$ associated with \eqref{eq2}, with
the $a_{ij}$ being the coefficients of the weighted adjacency matrix
$\mathcal{A}$. In terms of the coefficients of the associated
Laplacian matrix $L$, $\zeta_i$ can be rewritten as
\begin{equation}\label{zeta_l}
\zeta_i = \sum_{j=1}^{N}\ell_{ij}y_j.
\end{equation}
We refer to this as \emph{partial-state coupling} since only part of
the states are communicated over the network. When $C=I$, it means all states are communicated over the network and we call it \emph{full-state coupling}. Then, the original agents are expressed as
\begin{equation}\label{neq1}
\dot{x}_i=Ax_i+B\sigma(u_i)
\end{equation}
and $\zeta_i$ is rewritten as
\[
\zeta_i = \sum_{j=1}^{N}\ell_{ij}x_j.
\]

Obviously, state synchronization is achieved if
\begin{equation}\label{synch_org}
\lim_{t\to \infty} (x_i-x_j)=0.\quad \text{for all } i,j \in {1,...,N}
\end{equation}

For homogeneous MAS such as in this paper, almost all papers considered
state synchronization without imposing requirements on the
synchronized trajectory. However, for heterogenous agents, it has been
shown in \cite{wieland-sepulchre-allgower,grip-saberi-stoorvogel3}
that we basically need to consider regulated state synchronization
where the objective of the agents is to ensure that their state
asymptotically tracks a reference trajectory generated by a so-called exosystem.
Although we consider homogeneous MAS, we will study regulated state synchronization in this paper.

The reference trajectory is generated by the following exosystem
\begin{equation}\label{solu-cond}
\begin{system*}{rl}
\dot{x}_r & = A x_r\\
y_r&=Cx_r.
\end{system*}	
\end{equation}
with $x_r\in\R^n$.  Our objective is that the agents achieve regulated
state synchronization, that is
\begin{equation}\label{synchro}
\lim_{t\to \infty} (x_i-x_r)=0,
\end{equation}
for all $i\in\{1,\ldots,N\}$. Clearly, we need some level of
communication between the exosystem and the agents.  We assume that a nonempty
subset $\mathscr{C}$ of the agents have access to their
own output relative to the output of the exosystem.  Specially, each
agent $i$ has access to the quantity
\begin{equation}\label{elp}
\psi_i=\iota_{i}(y_i-y_r),\qquad
\iota_i=
\begin{cases}
1, & i\in \mathscr{C},\\
0, & i\notin \mathscr{C}.
\end{cases}
\end{equation}
By combining this with \eqref{zeta_l}, we have the following information
exchange
\begin{equation}\label{zeta_l-n}
\bar{\zeta}_i = \sum_{j=1}^{N}a_{ij}(y_i-y_j)+\iota_{i}(y_i-y_r).
\end{equation}
Meanwhile, for full-state coupling case  \eqref{zeta_l-n} will change as
\begin{equation}\label{zeta_l-nn}
\bar{\zeta}_i = \sum_{j=1}^{N}a_{ij}(x_i-x_j)+\iota_{i}(x_i-x_r).
\end{equation}
To guarantee that each agent can achieve the required regulation, we
need to make sure that there exists a path to each node starting with node from the
set $\mathscr{C}$. Therefore, we define the following set of graphs.
\begin{definition}\label{def_rootset}
	Given a node set $\mathscr{C}$, we denote by $\mathbb{G}_{\mathscr{C}}^N$ the set of all graphs with $N$ nodes containing the node set $\mathscr{C}$, such that every node of the network graph $\mathcal{G}\in\mathbb{G}_\mathscr{C}^N$ is a member of a directed tree
	which has its root contained in the node set $\mathscr{C}$.
\end{definition}
\begin{remark}
	Note that Definition \ref{def_rootset} does not require necessarily the existence of directed spanning tree.
\end{remark}

In the following, we will refer to the node set $\mathscr{C}$ as root set in view of Definiton \ref{def_rootset}.
For any graph $\mathbb{G}_\mathscr{C}^N$, with the Laplacian matrix $L$, we define the expanded Laplacian matrix as: 
\[
\tilde{L}=L+diag\{\iota_i\}=[\tilde{\ell}_{ij}]_{N \times N}
\]
which is not a regular Laplacian matrix associated to the graph, since the sum of its rows need not be zero. We know that Definition \ref{def_rootset}, guarantees that all the eigenvalues of $\tilde{L}$, have positive real parts. In particular matrix $\tilde{L}$ is invertible.

In this paper, we also introduce an additional
information exchange among protocols. In particular, each agent 
$i=1,\ldots, N$ has access to additional information, denoted by
$\hat{\zeta}_i$, of the form
\begin{equation}\label{eqa1}
\hat{\zeta}_i=\sum_{j=1}^Na_{ij}(\xi_i-\xi_j)
\end{equation}
where $\xi_j\in\mathbb{R}^n$ is a variable produced internally by agent $j$ and to be defined in next sections.

Now, we formulate the following problem for global regulated state synchronization of a MAS.

\begin{problem}\label{prob4}
	Consider a MAS described by \eqref{eq1} and \eqref{eq2} and the
	associated exosystem \eqref{solu-cond}. Let a set of nodes
	$\mathscr{C}$ be given which defines the set
	$\mathbb{G}_{\mathscr{C}}^N$. Let the associated network
	communication be given by \eqref{zeta_l-n}. 
	
	The \textbf{scalable global regulated state synchronization problem with additional information exchange} of a MAS is to find, if possible, a dynamic protocol for each agent $i\in\{1,\hdots,N\}$, using only
	knowledge of agent model, i.e. $(A,B,C)$, of the form:
	\begin{equation}\label{protoco5}
	\begin{system}{cl}
	\dot{x}_{c,i}&={f}(x_{c,i},\bar{\zeta}_i,\hat{\zeta}_i),\\
	u_i&={g}(x_{c,i}),
	\end{system}
	\end{equation}
	where $\hat{\zeta}_i$ is defined in \eqref{eqa1} with $\xi_i=H_{c}x_{i,c}$, and $x_{c,i}\in\R^{n_c}$, such that regulated
	state synchronization \eqref{synchro} is achieved for any $N$ and any
	graph $\mathcal{G}\in \mathbb{G}_{\mathscr{C}}^N$, and for all initial conditions of the agents $x_i(0) \in \mathbb{R}^n$, all initial conditions of the exosystem $x_r(0) \in \mathbb{R}^n$, and all initial conditions of the protocols $x_{c,i}(0) \in \mathbb{R}^{n_c}$.
\end{problem}

Next, we adopt semi-global framework to achieve regulated state synchronization by utilizing only \emph{linear} protocols.

\begin{problem}\label{prob3}
	Consider a MAS described by \eqref{eq1} and \eqref{eq2} and the
	associated exosystem \eqref{solu-cond}. Let a set of nodes
	$\mathscr{C}$ be given which defines the set
	$\mathbb{G}_{\mathscr{C}}^N$ and let the associated network
	communication be
	given by \eqref{zeta_l-n}. 
	
	The \textbf{scalable semi-global regulated state synchronization problem with additional information exchange} of a MAS is to find, if possible, a parametrized linear dynamic protocol with parameter $\eps\in (0,1]$ for each agent $i\in\{1,\hdots,N\}$, using only
	knowledge of agent model, i.e. $(A,B,C)$, of the form:
	
	\begin{equation}\label{protoco4}
	\begin{system}{cl}
	\dot{x}_{c,i}&=A_{c}^\eps x_{c,i}+B_{c}^\eps u_i+C_{c}^\eps \bar{\zeta}_i+D_{c}^\eps \hat{\zeta}_i,\\
	u_i&=F_c^{\eps}x_{c,i},
	\end{system}
	\end{equation}
	 where $\hat{\zeta}_i$ is defined in \eqref{eqa1} with $\xi_i=E_{c}x_{i,c}$, and  $x_{c,i}\in\R^{n_c}$, such that, for any given arbitrarily large compact sets $\mathbb{S}_a\in\mathbb{R}^n$, $\mathbb{S}_e\in\mathbb{R}^n$ and $\mathbb{S}_c\in\mathbb{R}^{n_c}$,  and for any $N$ and any graph $\mathcal{G}\in \mathbb{G}_{\mathscr{C}}^N$, there exists an $\eps^*$ such that, for all $\eps\in(0,\eps^*]$ regulated state synchronization \eqref{synchro} is achieved for all initial conditions of the agents in the set $\mathbb{S}_a$, all initial conditions of the exosystem in the set $\mathbb{S}_e$, and all initial conditions of the protocols in the set $\mathbb{S}_c$.
\end{problem}

\begin{remark}
	In the case of full-state coupling, matrix $C=I$ and we refer to Problems \ref{prob3} and \ref{prob4} with $\bar{\zeta}_i$ as \eqref{zeta_l-nn}, scalable semi-global and global regulated state synchronization problems for MAS with full-state coupling.
\end{remark}

\section{Scalable global regulated state synchronization of mas in presence of input saturation}

In this section, we will consider the scalable global regulated state synchronization problem for a
MAS with input saturation via scheduling (adaptive) design for both networks with full and partial-state coupling.
\subsection{Full-state coupling}
In this case, the following nonlinear protocol is designed for each agent
$i\in\{1,\ldots,N\}$,
\begin{equation}\label{pscp5}
\begin{system}{cll}
\dot{\chi}_i &=& A\chi_i+Bu_i+\bar{\zeta}_i-\hat{\zeta}_i-\iota_i\chi_i \\
u_i &=& -B\T P_{\eps(\chi_i)}\chi_i,
\end{system}
\end{equation}
where $P_\rho$ is the unique solution of  
\begin{equation}\label{arespecial2}
A\T P_\rho+ P_\rho A -  P_\rho BB\T
P_\rho + \rho P_\rho = 0 .
\end{equation}
and $\eps(\chi_i)$ is defined as
\begin{equation}\label{boundeps}
\eps(\chi_i)=\max\{\rho\in(0,1]:\chi_i\T P_{\rho}\chi_i\trace B\T P_{\rho}B\leq 1\}.
\end{equation}
Note that \cite{zhou-duan-lin} implies that $P_\rho$ is increasing in $\rho$ while $P_\rho \to 0$ as $\rho\to 0$. The agents communicate $\xi_i$ which is chosen as $\xi_i=\chi_i$, therefore each agent has access to the following information:
\begin{equation}\label{info1}
\hat{\zeta}_i=\sum_{j=1}^Na_{ij}(\chi_i-\chi_j).
\end{equation}
while $\bar{\zeta}_i$ is defined by \eqref{zeta_l-nn}.

Then, the synchronization result based on adaptation is stated in
Theorem \ref{mainthm2}. 

\begin{theorem}\label{mainthm2}
	Consider a MAS described by \eqref{neq1} satisfying Assumption \ref{Aass}, and the associated exosystem
	\eqref{solu-cond}. Let a set of nodes $\mathscr{C}$ be given which
	defines the set $\mathbb{G}_{\mathscr{C}}^N$. Let the associated
	network communication be given by \eqref{zeta_l-nn}.
	
	 Then, the scalable global
	regulated state synchronization problem as stated in Problem
	\ref{prob4} is solvable. In particular, the adaptive nonlinear
	dynamic protocol \eqref{pscp5}, \eqref{arespecial2}, and \eqref{boundeps} solves the regulated state
	synchronization problem for any $N$ and any graph
	$\mathcal{G}\in\mathbb{G}_{\mathscr{C}}^N$. 
\end{theorem}
%
%
%

\begin{proof}[Proof of Theorem \ref{mainthm2}]
Firstly, let $\tilde{x}_i=x_i-x_r$, we have
\[
\dot{\tilde{x}}_i=A{\tilde{x}}_i+B\sigma(u_i).
\]
Let $e_i=\tilde{x}_i-\chi_i$. According to \eqref{boundeps}, it yields by construction that $u_i$ does not get saturated, i.e., $\sigma(u_i)=u_i$, then we can obtain  
\begin{equation}
\begin{system*}{l}
\dot{\tilde{x}}_i=A\tilde{x}_i-BB\T P_{\eps(\tilde{x}_i-e_i)}(\tilde{x}_i-e_i),\\
\dot{e}_i=Ae_i-\sum_{j=1}^{N}\bar{\ell}_{ij}e_j.
\end{system*}	
\end{equation}
By defining 
	\[e=\begin{pmatrix}
	e_1\\\vdots\\e_N
\end{pmatrix} 
\]
we can obtain  
\begin{equation}\label{newsystem2}
\dot{e}=(I\otimes A-\bar{L}\otimes I)e
\end{equation}

Since all eigenvalues of $\bar{L}$ are positive, we have
\begin{equation}\label{boundapl}
(T\otimes I)(I\otimes A-\bar{L}\otimes I)(T^{-1}\otimes I)=I\otimes A-\bar{J}\otimes I
\end{equation}
for a non-singular transformation matrix $T$, where
\eqref{boundapl}  is upper triangular Jordan form with $A-\lambda_i I$ for $i=1,\cdots,N-1$ on the diagonal. Since all eigenvalues of $A$ are in the closed left half plane, $A-\lambda_i I$ is stable. Therefore, all eigenvalues of $I\otimes A-\bar{L}\otimes I$ have negative real part.
	Therefore, we have that the dynamics for $e_i$ are asymptotically stable.
	
%
%
	
	We choose the following Lyapunov function:
	\begin{equation}
	V_i=(\tilde{x}_i-e_i)\T P_{\eps_\alpha}(\tilde{x}_i-e_i)
	\end{equation}
	with $\eps_\alpha=\eps(\tilde{x}_i-e_i)$.
	
	Assume $V_i$ is non-increasing. Then we have
	\[
	\frac{dV_i}{dt}\leq0
	\]
	On the other hand, if $V_i$ is increasing then $\eps_\alpha$ is non-increasing, which implies that $P_{\eps_\alpha}$ is non-increasing. Meanwhile, we have
	\begin{align}
	\nonumber\frac{dV_i}{dt}\leq
	&(\tilde{x}_i-e_i)\T[P_{\eps_\alpha}(A-BB\T P_{\eps_\alpha})+(A-BB\T P_{\eps_\alpha})\T P_{\eps_\alpha}](\tilde{x}_i-e_i)\\
	\nonumber&-e_i\T[P_{\eps_\alpha}(A-BB\T P_{\eps_\alpha})+(A-BB\T P_{\eps_\alpha})\T P_{\eps_\alpha}]e_i\\
	\nonumber&+2e_i\T P_{\eps_\alpha}BB\T P_{\eps_\alpha} (\tilde{x}_i-e_i)-2(\tilde{x}_i-e_i)\T P_{\eps_\alpha}\dot{e}_i\\
	\nonumber&+(\tilde{x}_i-e_i)\T \frac{dP_{\eps_\alpha}}{dt}(\tilde{x}_i-e_i)\\
	\nonumber\leq &-\eps V_i +\|e_i\T[P_{\eps_\alpha}(A-BB\T P_{\eps_\alpha})+(A-BB\T P_{\eps_\alpha})\T P_{\eps_\alpha}]e_i\|\\
	&+2\left\| P_{\eps_\alpha}^{\frac{1}{2}}e_i\right\| \left\|P_{\eps_\alpha}^{\frac{1}{2}}BB\T P_{\eps_\alpha}^{\frac{1}{2}}\right\|V_i^{\frac{1}{2}}+2\left\|P_{\eps_\alpha}^{\frac{1}{2}}\dot{e}_i\right\|V_i^{\frac{1}{2}}\label{lypunbo}
	\end{align}
	with $P_{\eps_\alpha}$ satisfying \eqref{arespecial2}.
	
	Since $e_i$ is the state of an asymptotically stable system, there exist $z_1, z_2, z_3$, such that
	\begin{align*}
	&\|e_i\T[P_{\eps_\alpha}(A-BB\T P_{\eps_\alpha})+(A-BB\T P_{\eps_\alpha})\T P_{\eps_\alpha}]e_i\|_1\leq z_1\\
	&\left\| P_{\eps_\alpha}^{\frac{1}{2}}e_i\right\|_1 \leq z_2, 	\left\|P_{\eps_\alpha}^{\frac{1}{2}}\dot{e}_i\right\|_1\leq z_3
	\end{align*}
	
	Thus, we have  
	\[
	\frac{dV_i}{dt}\leq \beta_1(t)+\beta_2(t)V_i^{\frac{1}{2}}\leq(\beta_1(t) +\beta_2(t))(V_i+1)^{\frac{1}{2}}
	\]
	for suitable $\beta_1(t),\beta_2(t)\in \mathscr{L}_1$, and  $\beta_1(t),\beta_2(t)\ge 0$, and 
	\[
	\dot{W}(t)=(\beta_1 (t)+\beta_2(t))(W(t)+1)^{\frac{1}{2}}
	\]
	yields
	\[
	W(t)=\left[\int_0^t(\beta_1 (s)+\beta_2(s))ds+(W(0)+1)^{\frac{1}{2}}\right]^2-1
	\]
	Hence,
	\[
	V_i(t)\leq\left(\|\beta_1\|_1 +\|\beta_2\|_1+(V_i(0)+1)^{\frac{1}{2}}\right)^2
	\]
	Therefore $V_i(t)$ is bounded which implies $\eps_\alpha$ is bounded away from zero.
	
	Remains to show that $V_i\to 0$. From \cite[Lemma 6.1]{saberi-hou-stoorvogel}, we get
	\[
	\left|(\tilde{x}_i-e_i)\T \frac{dP_{\eps_\alpha}}{dt}(\tilde{x}_i-e_i)\right|\leq k\frac{dV_i}{dt}
	\]
	for some constant $k$. 
	
	Meanwhile, if $V_i$ is non-increasing we can derive, similar to our early analysis \eqref{lypunbo}, that
	\begin{align*}
	\frac{dV_i}{dt}\leq&-\eps_\alpha V_i+\bar{\beta}(V_i+1)^{\frac{1}{2}}+(\tilde{x}_i-e_i)\T \frac{dP_{\eps_\alpha}}{dt}(\tilde{x}_i-e_i)\\
	\leq&-\eps_\alpha V_i+\bar{\beta}(V_i+1)^{\frac{1}{2}}-k\frac{dV_i}{dt}
	\end{align*}
	where $\bar{\beta}=\beta_1+\beta_2$ and we get
	\[
	\frac{dV_i}{dt}\leq-\frac{\eps_\alpha}{1+k}V_i+\frac{\bar{\beta}}{1+k}(V_i+1)^{\frac{1}{2}}.
	\]
	
	But then we can prove that whether $V_i$ is increasing or decreasing we always have
	\[
	\frac{dV_i}{dt}\leq-\tilde{\alpha} V_i+\tilde{\beta}(V_i+1)^{\frac{1}{2}}.
	\]
	with $\tilde{\alpha}$ is a constant and lower bound of $\frac{\eps_\alpha}{1+k}$ and $\tilde{\beta}=\frac{\bar{\beta}}{1+k}\in \mathscr{L}_1$. Clearly, that implies $V_i\to 0$
\end{proof}\\

\subsection{Partial-state coupling}

For partial-state coupling, we design the following adaptive nonlinear protocol for each agent
$i\in\{1,\ldots,N\}$.
\begin{equation}\label{pscp4}
\begin{system}{cll}
\dot{\hat{x}}_i &=& A\hat{x}_i+B\hat{\zeta}_{i2}+K(\bar{\zeta}_i-C\hat{x}_i)+\iota_iBu_i \\
\dot{\chi}_i &=& A\chi_i+Bu_i+\hat{x}_i-\hat{\zeta}_{i1}-\iota_{i}\chi_i \\
u_i &=& -B\T P_{\eps(\chi_i)}\chi_i,
\end{system}
\end{equation}
where $K$ is a pre-design matrix such that $A-KC$ is Hurwitz stable. 
$P_{\eps(\chi_i)}$ is the unique solution of \eqref{arespecial2} with $\rho=\eps(\chi_i)$ where $\eps(\chi_i)$ is defined as \eqref{boundeps}. 
 In this protocol, the agents communicate $\xi_i=\begin{pmatrix}
\xi_{i1}\T,&\xi_{i2}\T
\end{pmatrix}\T=\begin{pmatrix}
\chi_i\T,&u_i\T
\end{pmatrix}\T$, i.e. each agent has access to additional information $\hat{\zeta}_i=\begin{pmatrix}
\hat{\zeta}_{i1}\T,&\hat{\zeta}_{i2}\T
\end{pmatrix}\T$, where:
\begin{equation}\label{add_1}
\hat{\zeta}_{i1}=\sum_{j=1}^Na_{ij}(\chi_i-\chi_j),
\end{equation}
and
\begin{equation}\label{add_2}
\hat{\zeta}_{i2}=\sum_{j=1}^{N}a_{ij}(u_i-u_j).
\end{equation}
while $\bar{\zeta}_i$ is defined via \eqref{zeta_l-n}.

Then, we obtain the synchronization result based on adaptation as the following theorem. 

\begin{theorem}\label{mainthm4}
	Consider a MAS described by \eqref{eq1} satisfying Assumption \ref{Aass}, and the associated exosystem
	\eqref{solu-cond}. Let a set of nodes $\mathscr{C}$ be given which
	defines the set $\mathbb{G}_{\mathscr{C}}^N$. Let the associated
	network communication be given by \eqref{zeta_l-n}. 
	
	Then, the scalable global regulated state synchronization problem as stated in Problem
	\ref{prob4} is solvable. In particular, the adaptive nonlinear
	dynamic protocol \eqref{pscp4}, \eqref{arespecial2}, and \eqref{boundeps} solves the scalable regulated state
	synchronization problem for any $N$ and any graph
	$\mathcal{G}\in\mathbb{G}_{\mathscr{C}}^N$. 
\end{theorem}

\begin{proof}[Proof of Theorem \ref{mainthm4}]
Similar to Theorem \ref{mainthm2}, let $\tilde{x}_i=x_i-x_r$.
We also define 
\[
\tilde{x}=\begin{pmatrix}
\tilde{x}_1\\\vdots\\\tilde{x}_N
\end{pmatrix}  
\hat{x}=\begin{pmatrix}
\hat{x}_1\\\vdots\\\hat{x}_N
\end{pmatrix}
\chi=\begin{pmatrix}
\chi_1\\\vdots\\\chi_N
\end{pmatrix}
\]
According to \eqref{boundeps}, it yields by construction that $u_i$ does not get saturated, i.e., $\sigma(u_i)=u_i$. By defining $e=\tilde{x}-\chi$ and $\bar{e}=(\bar{L}\otimes I)\tilde{x}-\hat{x}$, we can obtain  
\begin{equation}\label{newsystem}
\begin{system*}{l}
\dot{\tilde{x}}_i=A\tilde{x}_i-BB\T P_{\eps(\tilde{x}_i-e_i)}(\tilde{x}_i-e_i)\\
\dot{\bar{e}}=I\otimes (A-KC)\bar{e}\\
\dot{e}=(I\otimes A-\bar{L}\otimes I)e+\bar{e}
\end{system*}
\end{equation}
	Since all eigenvalues of $A-KC$ and $I\otimes A-\bar{L}\otimes I$ have negative real part, we obtain that dynamics of $e$ and $\bar{e}$ are all asymptotically stable.

	Then, we just need to prove the stability of 
	
	\[
		\dot{\tilde{x}}_i=A\tilde{x}_i-BB\T P_{\eps(\tilde{x}_i-e_i)}(\tilde{x}_i-e_i)
	\]
with $e_i$ and $\dot{e}_i$ in $\mathscr{L}_1$. Then, similar to the proof of Theorem \ref{mainthm2}, the synchronization result can be obtained.
\end{proof}

\section{Scalabale semi-global regulated state synchronization of mas in presence of input saturation}

In this section, we will consider the scalable semi-global regulated state synchronization problem for a
MAS with input saturation for networks with full- and partial-state coupling.

\subsection{Full-state coupling}

We will design a parametrized linear dynamic protocol with parameter
$\eps\in (0,1]$ for agent
$i\in\{1,\ldots,N\}$ as follows.
\begin{equation}\label{pscp1}
\begin{system}{cll}
\dot{\chi}_i &=& A\chi_i+Bu_i+\bar{\zeta}_i-\hat{\zeta}_i-\iota_{i}\chi_i \\
u_i &=& -B\T P_{\eps}\chi_i,
\end{system}
\end{equation}
where $P_{\eps}$ is the unique solution of the following ARE 
\begin{equation}\label{arespecial}
A\T P_{\eps} + P_{\eps} A -  P_{\eps} BB\T
P_{\eps} + \eps I = 0.
\end{equation}
Note that \cite{saberi-stoorvogel-sannuti-exter} implies that \eqref{arespecial} has a unique solution for any $\eps>0$ and $P_\rho \to 0$ as $\eps \to 0$.
The protocol requires the additional information $\hat{\zeta}_i$ as \eqref{info1}, while $\bar{\zeta}_i$ is defined by \eqref{zeta_l-nn}.

Our formal result is stated in the following theorem.
\begin{theorem}\label{mainthm1}
	Consider a MAS described by \eqref{neq1} satisfying Assumption \ref{Aass}, and the associated exosystem
	\eqref{solu-cond}. Let a set of nodes $\mathscr{C}$ be given which
	defines the set $\mathbb{G}_{\mathscr{C}}^N$. Let the associated
	network communication be given by \eqref{zeta_l-nn}.
	
	 Then, the scalable semi-global regulated state synchronization problem as stated in Problem
	\ref{prob3} is solvable. In particular, for any given compact sets $\mathbb{S}_a\in\mathbb{R}^n$, $\mathbb{S}_e\in\mathbb{R}^n$ and $\mathbb{S}_c\in\mathbb{R}^{n}$, and for any $N$ and any graph $\mathcal{G}\in\mathbb{G}_{\mathscr{C}}^N$, there exists $\eps^* > 0$ such that, for any $\eps\in(0,\eps^*]$, the dynamic protocol \eqref{pscp1} and \eqref{arespecial} solves the scalable regulated state
	synchronization problem.
\end{theorem}

\begin{proof}[Proof of Theorem \ref{mainthm1}]
	Firstly, let $\tilde{x}_i=x_i-x_r$, we have
	\[
	\dot{\tilde{x}}_i=A{\tilde{x}}_i+B\sigma(u_i)
	\]
	Then, we define 
	\[
	\tilde{x}=\begin{pmatrix}
	\tilde{x}_1\\\vdots\\\tilde{x}_N
	\end{pmatrix}  
	\chi=\begin{pmatrix}
	\chi_1\\\vdots\\\chi_N
	\end{pmatrix} \text{ and } 
	\sigma(u)=\begin{pmatrix}
	\sigma(u_1)\\\vdots\\\sigma(u_N)
	\end{pmatrix}
	\]
	then we have the following closed-loop system
	\begin{equation}
	\begin{system}{l}
	\dot{\tilde{x}}=(I\otimes A) \tilde{x}+(I\otimes B)\sigma(-I\otimes (B\T P_\eps)\chi)\\
	\dot{\chi}=(I\otimes (A-BB\T P_\eps) )\chi+(\bar{L}\otimes I)(\tilde{x}-\chi)
	\end{system}
	\end{equation}
	
	Let $e=\tilde{x}-\chi$, if the saturation is not active, we can obtain  
	\begin{equation}\label{newsystem6}
	\begin{system*}{l}
		\dot{\tilde{x}}=I\otimes (A-BB\T P_\eps) \tilde{x}+I\otimes (BB\T P_\eps)e\\
	\dot{e}=(I\otimes A-\bar{L}\otimes I)e
	\end{system*}
	\end{equation}
	Meanwhile from \eqref{boundapl}, we have that the dynamics for $e$ are asymptoticaly stable. Then, there exists an $\eps_1$ such that for $\eps<\eps_1$
	\begin{align}
	&\left\|(I\otimes B\T P_\eps)e\right\|_\infty<\frac{1}{2}\label{ebound1}\\
	&\left\|(I\otimes B\T P_\eps)e\right\|_2<1\label{ebound2}
	\end{align}
	since $e$ has asymptotically stable dynamics and bounded initial conditions.
	
	There exists $b_s>0$ and $\eps_2<\eps_1$ such that
	\begin{equation}\label{boundxt}
	\tilde{x}\T (I\otimes P_\eps) \tilde{x}<b_s+1
	\end{equation}
	implies that
	\begin{equation}\label{satux}
	\left\|(I\otimes B\T P_\eps)\tilde{x}\right\|_\infty<\frac{1}{2}
	\end{equation}
	for any $\eps<\eps_2$.
	
	There exists $\eps_3<\eps_2$ such that for $\eps<\eps_3$
	\begin{equation}\label{boundx0}
	\tilde{x}\T(0) (I\otimes P_\eps) \tilde{x}(0)<b_s
	\end{equation}
	for all initial conditions (inside a given compact set).
	
	Next, we prove \eqref{boundxt}. Consider the following Lyapunov function
	\begin{equation}
	V(t)=\tilde{x}\T (I\otimes P_\eps) \tilde{x},
	\end{equation}
	if the saturation is not active, then
	\begin{align*}
	\dot{V}(t)=&\tilde{x}\T (I\otimes\left[P_\eps(A-BB\T P_\eps)+(A-BB\T P_\eps)\T P_\eps\right])\tilde{x}\\
	&-2\tilde{x}\T\left[ I\otimes (P_\eps BB\T P_\eps)\right]e\\
	\leq &-\eps \tilde{x}\T\tilde{x}-\tilde{x}\T \left[I\otimes (P_\eps BB\T P_\eps)\right]\tilde{x}-2\tilde{x}\T\left[ I\otimes (P_\eps BB\T P_\eps)\right]e\\
	\leq & -\eps \tilde{x}\T\tilde{x}+e\T \left[I\otimes (P_\eps BB\T P_\eps)\right]e\\
	\leq &\|I\otimes (B\T P_\eps)e\|^2
	\end{align*}
	Integrating both sides of $\dot{V}(t)$, we have
	\[
	V(t)< V(0) + 1<b_s+1 
	\]
	using \eqref{ebound2}, which proves \eqref{boundxt} for $\eps\leq\eps^*$ with $\eps^*<\eps_3$.
	
	Inequality \eqref{boundxt} guarantees that \eqref{satux} is satisfied and by combining with \eqref{ebound1} we find that the saturation never gets activated, i.e.
	\begin{equation}\label{uas}
	\sigma(-(I\otimes B\T P_{\eps})(\tilde{x}-e))=-(I\otimes B\T P_{\eps})(\tilde{x}-e)
	\end{equation}
	
	Then, since we have $I\otimes A-\bar{L}\otimes I$ is asymptotically stable from \eqref{boundapl}, we just need to prove the stability of
	\begin{equation}\label{statefeedback3}
	\dot{\tilde{x}}=I\otimes (A-BB\T P_\eps) \tilde{x}
	\end{equation}
	which $A-BB\T P_\eps$ is Hurwitz stable. Therefore, we can obtain the asymptotical stability of \eqref{newsystem2}, i.e.,
	\[
	\lim_{t\to \infty}\tilde{x}_i\to 0.
	\]
	It implies that $x_i-x_r\to0$, which proves the result.
\end{proof}

\subsection{Partial-state coupling}

Now, we consider the case via partial-state coupling.
We design a parametrized linear dynamic protocol with parameter
$\eps\in (0,1]$ for agent
$i\in\{1,\ldots,N\}$ as follows.
\begin{equation}\label{pscp3}
\begin{system}{cll}
\dot{\hat{x}}_i &=& A\hat{x}_i+B\hat{\zeta}_{i2}+K(\bar{\zeta}_i-C\hat{x}_i)+\iota_i Bu_i \\
\dot{\chi}_i &=& A\chi_i+Bu_i+\hat{x}_i-\hat{\zeta}_{i1}-\iota_{i}\chi_i \\
u_i &=& -B\T P_{\eps}\chi_i,
\end{system}
\end{equation}
where $K$ is a pre-design matrix such that $A-KC$ is Hurwitz stable and $P_{\eps}$ is the unique solutions of \eqref{arespecial}.
Moreover, the agents communicate $\xi_i=\begin{pmatrix}
\xi_{i1}\T,&\xi_{i2}\T
\end{pmatrix}\T=\begin{pmatrix}
\chi_i\T,&u_i\T
\end{pmatrix}\T$, i.e. each agent has access to additional information $\hat{\zeta}_i=\begin{pmatrix}
\hat{\zeta}_{i2}\T,&\hat{\zeta}_{i2}\T
\end{pmatrix}\T$ as \eqref{add_1} and \eqref{add_2}, while $\bar{\zeta}_i$ is defined via \eqref{zeta_l-n}.

Then we have the following theorem for MAS via partial-state coupling.

\begin{theorem}\label{mainthm3}
	Consider a MAS described by \eqref{eq1} satisfying Assumption \ref{Aass}, and the associated exosystem
	\eqref{solu-cond}. Let a set of nodes $\mathscr{C}$ be given which
	defines the set $\mathbb{G}_{\mathscr{C}}^N$. Let the associated
	network communication be given by \eqref{zeta_l-n}. 
	
	Then, the scalable semi-global regulated state synchronization problem as stated in Problem
	\ref{prob3} is solvable. In particular, for any given compact sets $\mathbb{S}_a\in\R^n$, $\mathbb{S}_e\in\R^n$, and $\mathbb{S}_c\in\R^{2n}$, and for any $N$ and any graph
	$\mathcal{G}\in\mathbb{G}_{\mathscr{C}}^N$, there exists $\eps^* > 0$ such that, for any $\eps\in(0,\eps^*]$ the dynamic protocol \eqref{pscp3} and \eqref{arespecial} solves the scalable regulated state
	synchronization problem.
\end{theorem}

\begin{proof}[Proof of Theorem \ref{mainthm3}]
	Similar to Theorem \ref{mainthm1}, let $\tilde{x}_i=x_i-x_r$, we have
	\[
	\begin{system}{cll}
	\dot{\tilde{x}}_i&=&A{\tilde{x}}_i+B\sigma(u_i)\\
	\dot{\hat{x}}_i &=& A\hat{x}_i+B\hat{\zeta}_{i2}+K(\bar{\zeta}_i-C\hat{x}_i)+\iota_i Bu_i\\
	\dot{\chi}_i &=& A\chi_i+Bu_i+\hat{x}_i-\hat{\zeta}_{i1}-\iota_{i}\chi_i 
	\end{system}
	\]
	
	we define 
	\[
	\tilde{x}=\begin{pmatrix}
	\tilde{x}_1\\\vdots\\\tilde{x}_N
	\end{pmatrix}  
	\hat{x}=\begin{pmatrix}
	\hat{x}_1\\\vdots\\\hat{x}_N
	\end{pmatrix}
	\chi=\begin{pmatrix}
	\chi_1\\\vdots\\\chi_N
	\end{pmatrix} \text{ and } 
	\sigma(u)=\begin{pmatrix}
	\sigma(u_1)\\\vdots\\\sigma(u_N)
	\end{pmatrix}
	\]
	then we have the following closed-loop system
	\begin{equation}
	\begin{system*}{l}
	\dot{\tilde{x}}=(I\otimes A) \tilde{x}+(I\otimes B)\sigma(-I\otimes (B\T P_\eps)\chi)\\
	\dot{\hat{x}} = I\otimes (A-KC)\hat{x}-(\bar{L}\otimes BB\T P_\eps)\chi+(\bar{L}\otimes KC)\tilde{x} \\
	\dot{\chi} = (I\otimes A-\bar{L}\otimes I)\chi-(I\otimes BB\T P_\eps)\chi+\hat{x}
	\end{system*}
	\end{equation}
	
	Then, similar to the proof (from \eqref{ebound1} to \eqref{uas}) of Theorem \ref{mainthm1}, we have that the saturation never gets activated for $\eps\leq\eps^*$ with $\eps^*\ll1$. By defining $e=\tilde{x}-\chi$ and $\bar{e}=(\bar{L}\otimes I)\tilde{x}-\hat{x}$, we can obtain  
	\begin{equation}\label{newsystem3}
	\begin{system*}{l}
	\dot{\tilde{x}}=I\otimes (A-BB\T P_\eps) \tilde{x}+I\otimes (BB\T P_\eps)e\\
	\dot{\bar{e}}=I\otimes (A-KC)\bar{e}\\
	\dot{e}=(I\otimes A-\bar{L}\otimes I)e+\bar{e}
	\end{system*}
	\end{equation}

	Since $I\otimes A-\bar{L}\otimes I$, $A-KC$, and $A-BB\T P_\eps$ are stable, we obtain 
	\[
	\lim_{t\to \infty}\tilde{x}_i\to 0
	\]
	It implies that $x_i-x_r\to0$, which proves the result.
\end{proof}

\section{Numerical Example}
In this section, we will illustrate the effectiveness of our protocols with a numerical example for global synchronization of MAS with partial-state coupling. We show that our one shot protocol design \eqref{pscp4} works for any graph with any number of agents.  

Consider the agents model \eqref{eq1} as:
\begin{equation*}\label{ex}
\begin{cases}
\dot{x}_i=\begin{pmatrix}
0&1&0\\0&0&1\\0&0&0
\end{pmatrix}x_i+\begin{pmatrix}
0\\0\\1
\end{pmatrix}\sigma(u_i),\\
y_i=\begin{pmatrix}
1&0&0
\end{pmatrix}x_i
\end{cases}
\end{equation*}
and the exosystem:
\begin{equation*}\label{exo_ex}
\begin{cases}
\dot{x}_r=\begin{pmatrix}
0&1&0\\0&0&1\\0&0&0
\end{pmatrix}x_r,\\
y_r=\begin{pmatrix}
1&0&0
\end{pmatrix}x_r
\end{cases}
\end{equation*}

We consider three different MAS with different number of agents and different communication topologies to show that the designed protocol is independent of the
communication network and number of agents $N$. 

\begin{itemize}
	\item \emph{Case $1$:} Consider MAS with $5$ agents $N=5$, and directed communication topology shown in Figure \ref{graph_1}.\\

 \begin{figure}[h]
	\includegraphics[width=5cm, height=2.5cm]{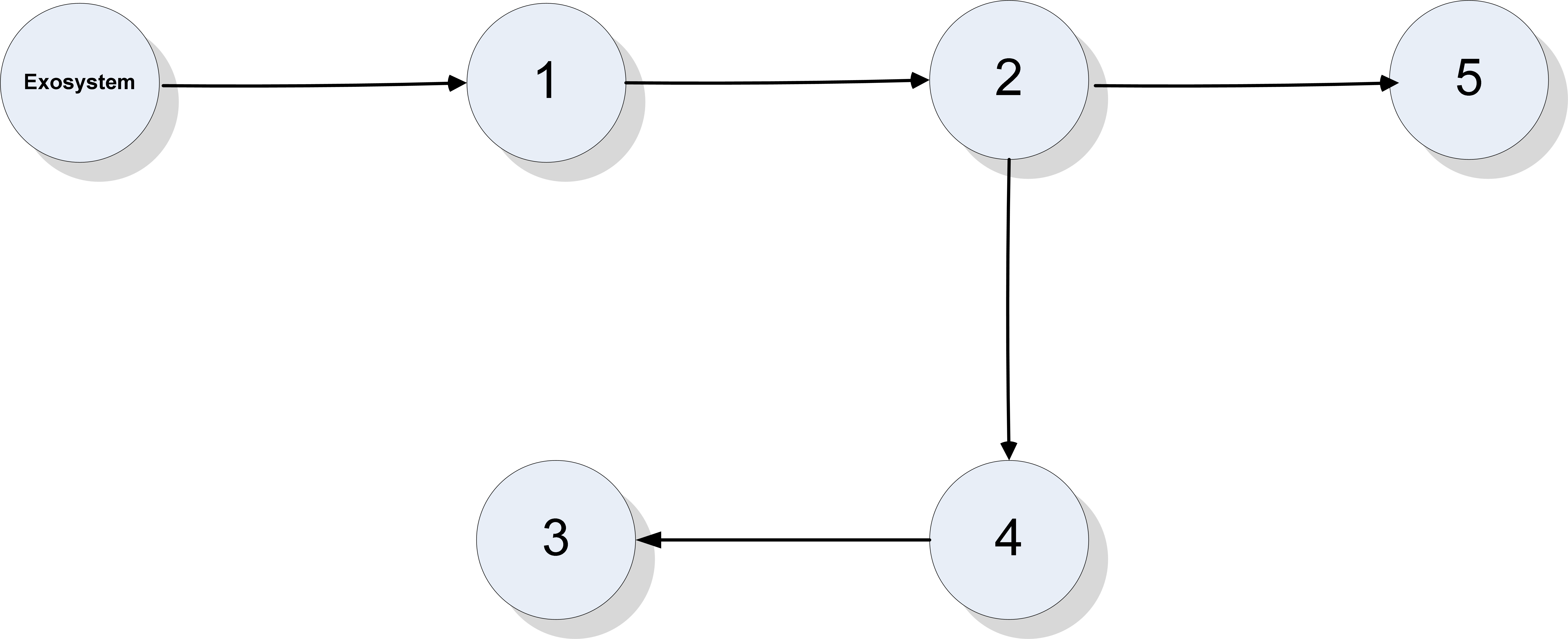}
	\centering
	\vspace*{-.2cm}
	\caption{The directed communication network, Case $1$}\label{graph_1}
\end{figure}

	\item \emph{Case $2$:} In this case, we consider MAS with $6$ agents $N = 6$, and directed communication topology shown in Figure \ref{graph_2}.

\begin{figure}[h!]
	\includegraphics[width=5cm, height=3cm]{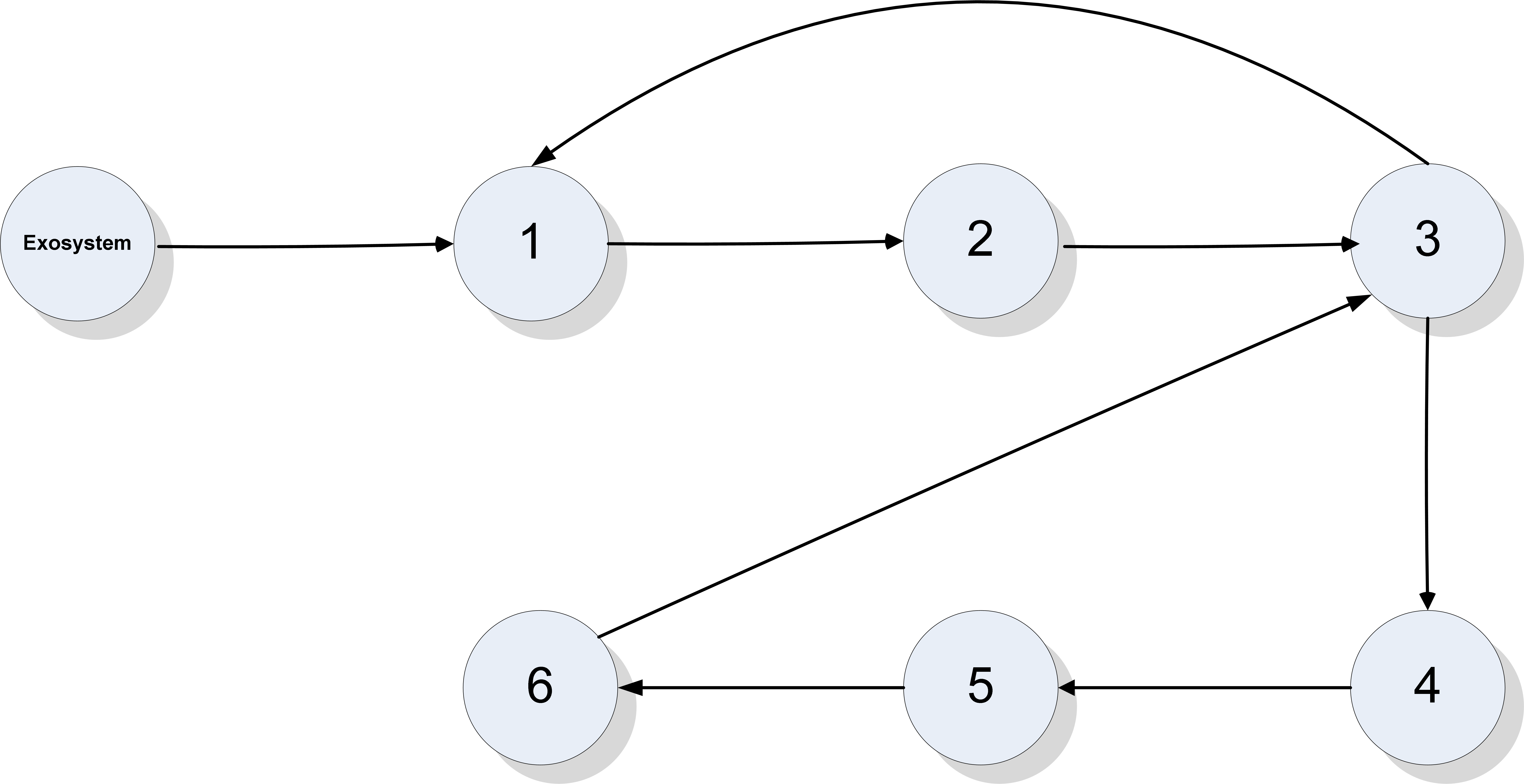}
	\centering
	\vspace*{-.2cm}
	\caption{The directed communication network, Case $2$}\label{graph_2}
\end{figure}

	\item \emph{Case $3$:} Finally, we consider the MAS with $3$ agents, $N =3$ and communication graph shown in Figure \ref{graph_3}.

\begin{figure}[h!]
	\includegraphics[width=4cm, height=2.5cm]{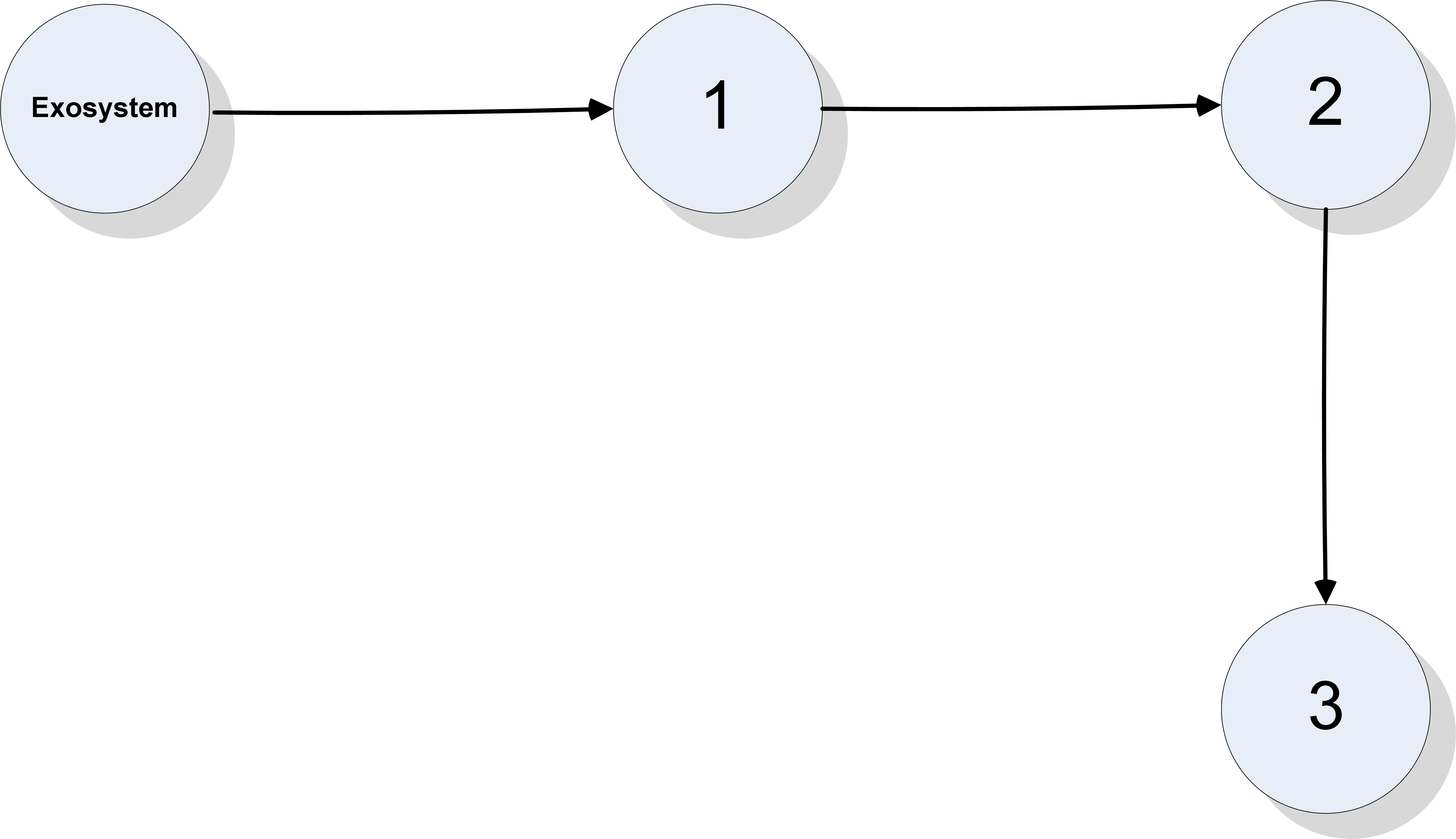}
	\centering
	\vspace*{-.2cm}
	\caption{The directed communication network, Case $3$}\label{graph_3}
\end{figure}
\end{itemize}
The results are demonstrated in Figure \ref{results_case11}-\ref{results_case33}. The simulation results show that the protocol design is independent of the communication graph and is scale free so that we can achieve state synchronization with one shot protocol design as \eqref{pscp4}, for any graph with any number of agents. 

\begin{figure}[t]
	\includegraphics[width=9cm, height=8cm]{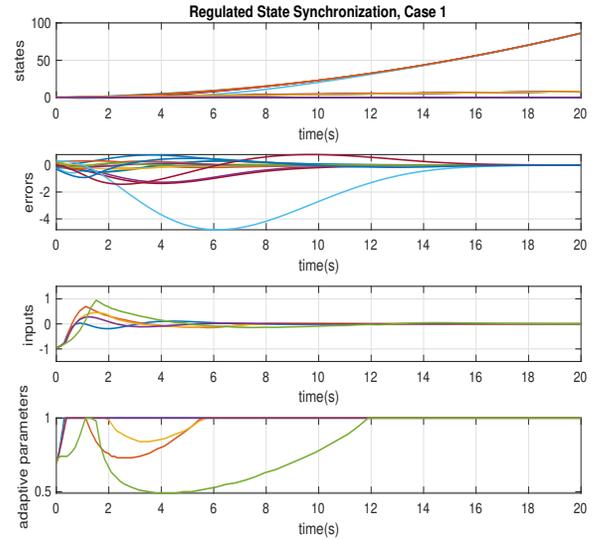}
	\centering
	\vspace*{-1cm}
	\caption{Results for MAS with associated communication graph $1$}\label{results_case11}
\end{figure}

\begin{figure}[t]
	\includegraphics[width=9cm, height=8cm]{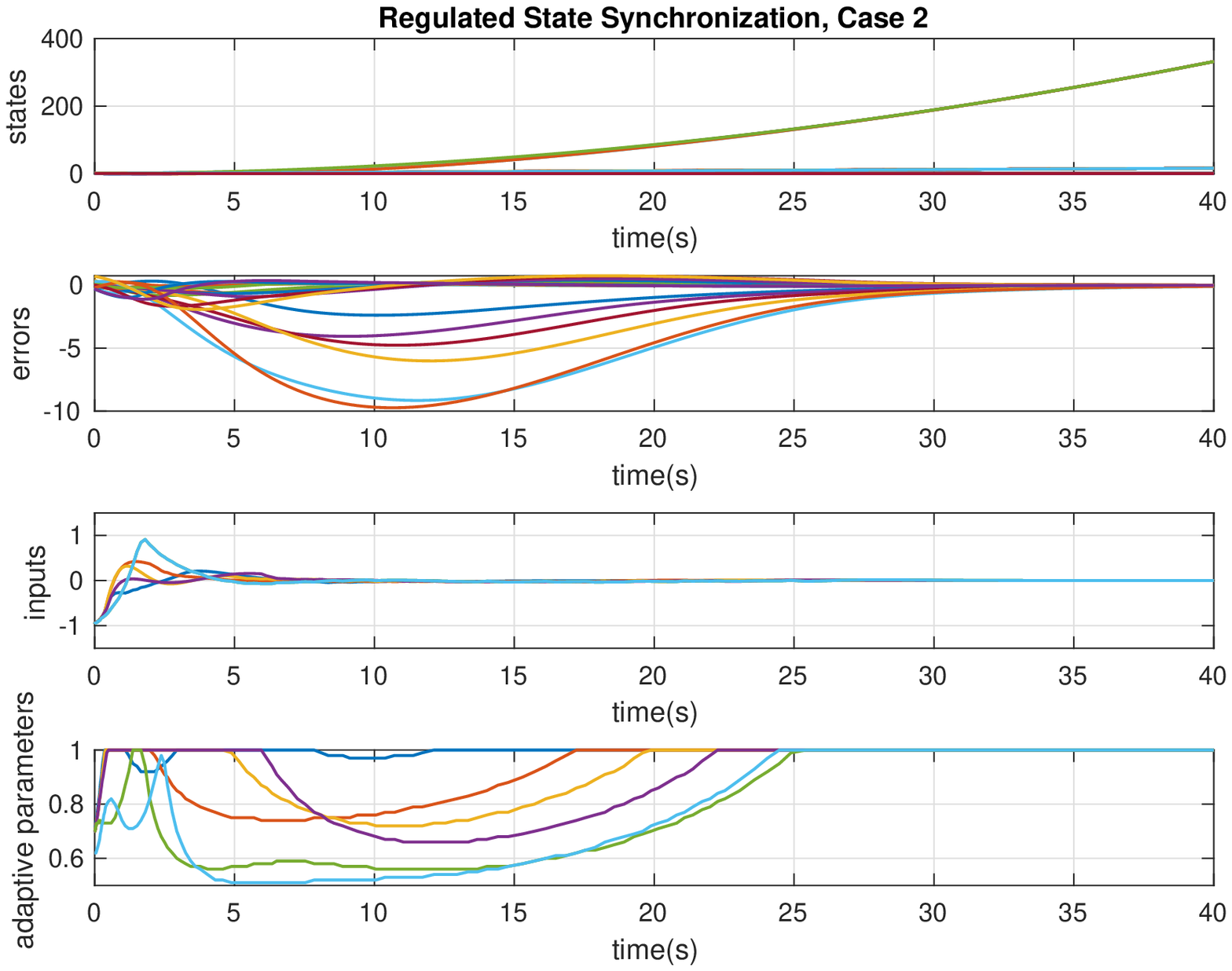}
	\centering
	\vspace*{-1cm}
	\caption{Results for MAS with associated communication graph $2$}\label{results_case22}
\end{figure}

\begin{figure}[t!]
	\includegraphics[width=9cm, height=8cm]{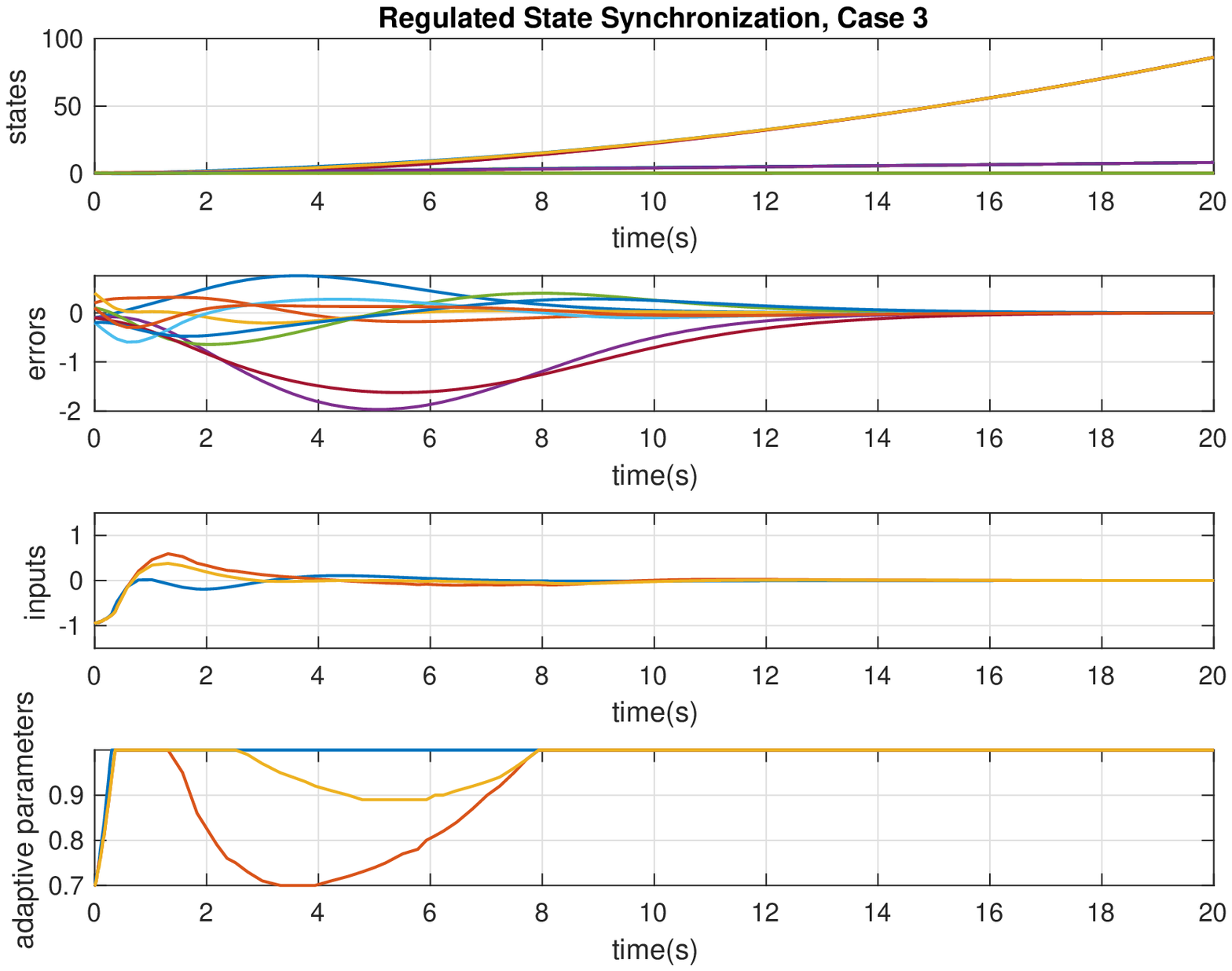}
	\centering
	\vspace*{-1cm}
	\caption{Results for MAS with associated communication graph $3$}\label{results_case33}
\end{figure}

\bibliographystyle{plain}
\bibliography{referenc}
\end{document}